\documentclass[journal]{IEEEtran}
\usepackage{amsfonts}
\usepackage{bbm}
\usepackage[utf8]{inputenc}
\usepackage{epstopdf} %converting to PDF
\usepackage{epsfig}
\DeclareGraphicsExtensions{.eps}
\usepackage{tikz}
\usetikzlibrary{matrix,shapes,arrows,positioning,chains}
\newcommand{\boldm}[1] {\mathversion{bold}#1\mathversion{normal}}
\usepackage{amsmath, amsthm, amssymb} 
\newtheorem{theorem}{Theorem}
\newtheorem{corollary}{Corollary}
\newtheorem{lemma}{Lemma}
\newtheorem{prop}{Proposition}
\newtheorem{assumption}{Assumption}
\usepackage{cite}
\usepackage{subfigure}
\usepackage{graphicx}
\usepackage[export]{adjustbox}
\usepackage{hyperref}
\usepackage{nopageno}
\usepackage{changes}

\ifCLASSINFOpdf
\else
\fi

\begin{document}

\title{Performance Analysis of Joint Pairing and Mode Selection in D2D Communications with FD Radios}

%\author{Simin~Badri\IEEEauthorrefmark{1},
%        Mansour~Naslcheraghi\IEEEauthorrefmark{1},
%        and~Mehdi~Rasti\IEEEauthorrefmark{2},% <-this % stops a space

\author{Simin~Badri,
	Mansour~Naslcheraghi,
	and~Mehdi~Rasti% <-this % stops a space
        
%                \IEEEauthorblockA{
%        	Email: \IEEEauthorrefmark{1}\{s.badri, m.naslcheraghi\}@ieee.org,
%        	\IEEEauthorrefmark{2}rasti@aut.ac.ir}
%        \IEEEauthorblockA{Department of Computer Engineering and IT,
%        	Amirkabir University of Technology, Tehran, Iran \\
%        	Email: \IEEEauthorrefmark{1}\{s.badri, m.naslcheraghi\}@ieee.org,
%        \IEEEauthorrefmark{2}rasti@aut.ac.ir}
\IEEEcompsocitemizethanks{\IEEEcompsocthanksitem The authors are with the Department of Computer and IT Engineering, Amirkabir University of Technology, Tehran, Iran. Email:\{s.badri, m.naslcheraghi\}@ieee.org,
	        	rasti@aut.ac.ir\protect\\}
}% <-this % stops a space

\maketitle
\thispagestyle{empty}
\begin{abstract}
In cellular-D2D networks, users can select the communication mode either direct and form D2D links or indirect and communicate with BS. In former case, users should perform pairing selection and choose their pairs. The main focus in this paper is proposing an analytical framework by using tools from stochastic geometry to address these two issues, i.e. i) mode selection for the user devices to be established in either cellular or D2D mode, which is done based on received power from BS influenced by a bias factor, and ii) investigation of choosing $n^{\rm{th}}$-nearest neighbor as the serving node for the receiver of interest, by considering full-duplex (FD) radios as well as half-duplex (HD) in the D2D links. The analytic and simulation results demonstrate that even though the bias factor determines the throughput of each mode, it does not have any influence on the system sum throughput. Furthermore, we demonstrate that despite of suffering from self-interference, FD-D2D results in higher system sum throughput as well as higher coverage probability in comparison to its counterpart, namely purely HD-D2D network.
\end{abstract}


\IEEEpeerreviewmaketitle

\section{Introduction}\label{introduction}
Tremendous growing demand in mobile data traffic is one of the most important challenges in future cellular networks. Due to the spectrum scarcity, the cellular infrastructure confronted with major challenges providing efficient methods of resource utilization in order to accommodate exploding demands. Direct D2D and In-Band-Full-Duplex (IBFD) communications are proposed as two promising technologies in 5G in order to enhance spectrum efficiency. By reusing the spectrum of cellular users, D2D users bypass the BS and form direct link; thus the spectral efficiency can be improved. In addition to improving spectral efficiency by using D2D communication, IBFD can increase the spectral efficiency and throughput nearly up to double. FD enables two devices to transmit and receive in a single frequency band at the time and thus improves the attainable spectral efficiency. The idea of exploiting IBFD was raised in recent years, however, due to the presence of self-interference (SI), it was not considered to be applicable. Nevertheless, with the advancements of technology in analog and digital SI cancellation techniques as well as the antenna designs, the SI issue is now significantly resolved \cite{A}. To be practical, in addition to exploiting SI cancellation, FD communication is considered to be applied in short range distances because of less required transmit power and SI. Due to having short range distances, D2D communication can be used with FD to further improve the network throughput.

In order to efficiently exploit the potentials of D2D, there are some challenges that should be properly addressed. Two of the most important ones are cellular-D2D mode selection and pairing. In former, users can choose between cellular and D2D mode and in the latter users should select their pairs for communication in D2D mode. In studies \cite{B,C,D}, the mode selection between cellular and D2D communication is investigated for uplink transmission in a HD network. The throughput for a FD network is analyzed in \cite{E, F, G}. In all these works, the location of receiver in a D2D pair is modeled by either uniform distribution, fixed distance or nearest neighbor distribution. However, one of the fundamental requirements of D2D networks is the possibility of having multiple proximate devices as a serving device that neither of mentioned works considered this. In \cite{H} authors investigate the availability of content in $k$th closest device to a given device where the location of D2D devices is modeled as Poisson cluster process (PCP) which suits for finite networks.

The main focus of this work is considering the aforementioned issues namely, joint cellular-D2D mode selection and D2D pairing, and to the best of our knowledge, there is no existing work in the literature to provide performance analysis for joint cellular-D2D mode selection and pairing. Details along with the main contributions are as follows, 

\begin{itemize}
	\item The proposed system contains two different phases; (i) Cellular/D2D mode selection in which nearest BS received power is considered as the selection criteria. In particular, a user device chooses cellular mode if the biased nearest BS's received power can meet some predefined system threshold. Different from the studies \cite{B} and \cite{D}, in which cellular uplink spectrum is utilized for the mode selection, in this work we utilize cellular downlink spectrum for cellular/D2D mode selection. Also, these works consider only cellular/D2D mode selection for a HD network while we perform both cellular/D2D mode selection as well as pairing for a hybrid HD/FD network.
	
	\item We propose a D2D pairing scheme, in which that pairing is done based on $n^{\rm{th}}$-nearest neighbor, which is more challenging where we aim to develop analysis for an infinite network, i.e. Poisson Point Process (\textbf{PPP}). 	
\end{itemize} 

The remainder of this paper are organized as follows. In Section \ref{System Model}, the system model, assumptions, mode selection and pairing schemes are described. The coverage probability and spectral efficiency for both cellular mode and D2D mode are analyzed in section \ref{CovProb_Th}. The results and discussion of simulation are provided in Section \ref{results}, and Section \ref{Conclusion} concludes this work.

\section{System Model}\label{System Model}
Consider a cellular-D2D overlaied network and assume downlink spectrum for cellular communications. The BSs are modeled by an independent \textbf{PPP} {\boldm $\Phi_{\rm b}$}$=\{b_i; i = 1,2,3,... \}$ with intensity of $\lambda_{\rm b}$  in $\mathbb{R}^2$. The location of user equipment (UE) follows an independent \textbf{PPP} {\boldm $\Phi_{\rm u}$}$=\{u_i;i = 1,2,3,...\}$ with intensity of $\lambda_{\rm u}$ in $\mathbb{R}^2$. We perform the analysis based on a typical UE, which is cellular UE (CUE) for downlink cellular communication and D2D UE (DUE) for D2D link. We use a general power-law path loss model with path loss exponent $\alpha>2$. All BSs and UEs choosing D2D mode will transmit with fixed power $P_{\rm b}$ and $P_{\rm d}$, respectively. We consider Rayleigh fading, i.e, the channel gain between two points $x, y \in \mathbb{R}^2$, $h(x,y)$ which is exponentially distributed with mean one denoted by $h(x,y) \sim \exp(1)$. Fig. \ref{Topology} illustrates the system model as explained in what follows.

\begin{figure}
	\centering
	\includegraphics[width=0.5 \textwidth]{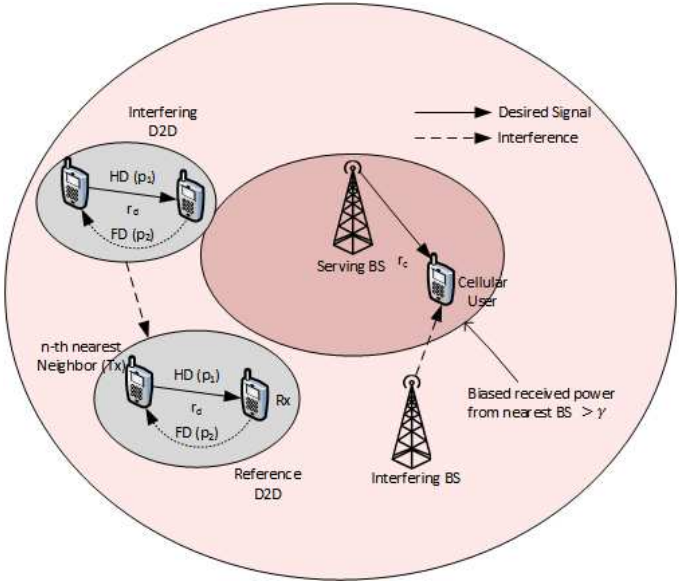}
	\caption{Illustration of the system model.}{\vspace{-4mm}}
	\label{Topology}
\end{figure}

\subsection{Cellular/D2D Mode Selection}
Mode selection for UEs is performed based on biased received power, in which a UE connects to its nearest BS if the received power from this BS exceeds a predefined threshold, $\gamma$. Specifically, a UE associates to its nearest BS if 
\begin{equation}\label{User_association_condition}
P_{r_c} = k P_{\rm b} h_c r_c^{-\alpha} > \gamma,
\end{equation}
where $P_{\rm b}$, $h_{\rm c}$, $r_{\rm c}$ and $\alpha$ are the transmit power of nearest BS, the power fading at distance  
$r_{\rm c}$, the distance between a node and its nearest BS and the path loss exponent, respectively. $k$ is a bias parameter used to tune the density of CUEs and DUEs across the entire network, i.e. $k \to 0$ means that each UE is forced to choose D2D mode. By adjusting $ k \to \infty$, each UE has this chance to associate with its nearest BS. This tuning parameter has influence on data offloading and resource allocation. That is by setting $k$ to larger values most UEs tend to D2D communication and therefore offload cellular traffic and free the resources to cover more CUEs. On the other side, setting $k$ to lower values makes it less probable for UEs to choose D2D mode as the communication link and thereby mitigates the throughput of D2D network. In cellular case, the distance between the typical UE and the $i^{th}$ BS is denoted with $r_{{\rm c}i}$. Among these distances, we denote the distance between the typical UE and its nearest BS with $r_{\rm c}$. It is shown that the random variable $r_{\rm c}$ is Rayleigh distributed \cite{I} and the corresponding probability density function (PDF) is given by
\begin{equation}\label{PDF_r_c}
	f_{r_{\rm c}}(r) = 2\pi \lambda_{\rm b} r e^{-\pi \lambda_{\rm b} r^2}.
\end{equation}

\subsection{Pairing}
If after using a bias factor mode selection explained in previous subsection, a given UE looses the option connecting to its nearest BS, we assume it connects to the $n^{\rm{th}}$-nearest neighbor and forms a D2D link. Although D2D discovery and its mechanisms could be considered as a part of pairing, we leave it as our future work. However, various D2D discovery mechanisms can be adopted to discover D2D pairs such as what have been suggested in \cite{J}. Let us denote the distance between the typical UE and its transmitter by $r_{\rm d}$. In this regard,
the PDF of the distance to $n^{\rm{th}}$-nearest UE, $r_{\rm d}$, in two dimensions is given by \cite{I},
\begin{equation}\label{PDF_r_d}
g_{(n)}(r_{\rm d}) = \frac{2}{\Gamma (n)} (\lambda_{\rm d} \pi)^n r_{\rm d}^{2n - 1} \exp(-\pi \lambda_{\rm d} r_d^2).
\end{equation}

It is worth mentioning that unlike other papers' assumptions that consider the D2D UEs in a pair to be uniformly distributed in the vicinity of each other such as \cite{B} or are the nearest neighbors to each other like \cite{D}, our pairing scheme is more general and flexible since by setting $n$ to one we achieve nearest neighbor selection scheme. This scheme can be used in variety of applications providing proximity services introduced by Long Term Evolution (LTE) advanced.

\section{Coverage Probability and Spectral Efficiency}\label{CovProb_Th}

For a typical UE located at distance $x$ from its transmitter, and $\beta$ as target signal-to-interference-plus-noise ratio (SINR), the coverage probability of this UE is given by
\begin{equation}\label{Formula: Coverage Probe}
\mathcal{C}_{\mathcal{X}} = \mathbb{E}_x[\mathbb{P}[\text{SINR}_{\mathcal{X}}(x) \geq \beta]].
\end{equation}
Eq. (\ref{Formula: Coverage Probe}) holds for any arbitrary node operating in mode $\mathcal{X} \in \{{\rm c},{\rm d}\}$, where ${\rm c}$ and ${\rm d}$ represent the cellular and D2D modes, respectively. The definition of SINR for cellular and D2D mode is given in subsections \ref{Cellular_mode} and \ref{D2D_mode}, respectively. The average rate of a typical UE is given by,
\begin{equation}\label{Ergodic_rate_definition}
\mathcal{R}_{\mathcal{X}} \triangleq \mathbb{E}_x[\mathbb{E}_{\text{SINR}_{\mathcal{X}}}[\text{ln}(1+\text{SINR}_{\mathcal{X}}(x)]].
\end{equation}
In Lemma \ref{lemma_BS_association_probability} we obtain BS association probability needed to derive coverage probability and spectral efficiency.
\begin{lemma}\label{lemma_BS_association_probability}
	\textit{The probability that a {\rm UE} associates to its nearest {\rm BS} is}
	\begin{flalign}\label{BS_association_probability}
	& {\mathcal{P}} = 2\pi {\lambda_{\rm b}}\int_0^\infty  {\exp \Big(- \gamma {{(k{P_{\rm b}})}^{-1}}{r_{\rm c}^\alpha } - \pi {\lambda _{\rm b}}{r_{\rm c}^2}\Big) r_{\rm c} dr_{\rm c}} & \\\notag & \stackrel{(\alpha = 4)}{=}
	\pi \lambda_{\rm b} \sqrt{\frac{\pi k P_{\rm b}}{4 \gamma}} \exp \Big(\frac{k P_{\rm b} \pi^2 \lambda_{\rm b}^2}{4 \gamma}\Big) \Big [1 - \Phi (\pi \lambda_{\rm b} \sqrt{\frac{k P_{\rm b}}{4 \gamma}})\Big],
	\end{flalign}
	where $\Phi$ is the probability integral\footnote{Probability integral is defined by $\Phi(x) = \frac{1}{\sqrt{2 \pi}} \int_{0}^x e^{-t^2}dt$}.
\end{lemma}
\begin{proof}
	When $P_{r_{\rm c}} > \gamma$, the typical {\rm UE} associates to its nearest {\rm BS} and becomes a {\rm CUE}. Therefore,
	\begin{flalign}\label{Proof_BS_association_probability}
	\mathcal{P} & = \mathbb{P}[P_{r_{\rm c}} > \gamma] = \mathbb{P}[h > \gamma (k P_{\rm b})^{-1} r_{\rm c}^{\alpha}] \notag & \\ & \stackrel{(a)}{=}
	2 \pi \lambda_{\rm b} \int_0^{\infty}e^{-\gamma (k P_{\rm b})^{-1} r_{\rm c}^{\alpha}} e^{- \pi \lambda_{\rm b} r_{\rm c}^2}r_{\rm c} {\rm d}r_{\rm c} \notag & \\ & =
	2 \pi \lambda_{\rm b} \int_0^{\infty} \exp\Big(-\gamma (k P_{\rm b})^{-1} r_{\rm c}^{\alpha} - \pi \lambda_{\rm b} r_{\rm c}^2 \Big)r_{\rm c} {\rm d}r_{\rm c},
	\end{flalign}
	where ($a$) follows from the equation (\ref{PDF_r_c}).
\end{proof}
From Lemma 1, we observe that BS association is influenced by BSs' density, transmit power and bias factor. The effect of these parameters are intuitive, that is UEs are more favorable to associate to a BS and thereby choose cellular mode when the number of BSs or transmit power increases. Now, having the densities of the CUEs ($\lambda_{\rm c} = \lambda_{\rm u} \mathcal{P}$) and DUEs ($\lambda_{\rm d} = \lambda_{\rm u} (1 - \mathcal{P})$) ,we aim to analyze coverage probability and average rate for both cellular and D2D modes.

\subsection{Cellular mode}\label{Cellular_mode}
In an overlaid D2D network described in section \ref{System Model}, CUEs suffer from the inter-cell interference caused by the neighboring BSs. Therefore, the experienced SINR at a typical CUE located in distance $r_{\rm c}$ from its associated BS is given by
\begin{align}
\text{SINR}_{{\rm c}}(r_{\rm c}) = \frac{P_{\rm b} h_{\rm c} r_{\rm c}^{-\alpha}}{\sum_{i \in \text{{\boldm $\Phi_{\rm b}$}} \backslash \{o\}} P_{\rm b} h_{i} r_{{\rm c}i}^{-\alpha} + \sigma ^2}.
\end{align}
where $\sigma ^2$ is the background noise.
It is worth to note that even though the PDF of the distance between a typical UE and its nearest BS follows Rayleigh distribution, here, the assumption of being in cellular mode affects the aforementioned PDF and we have to compute the PDF of the distance between CUE and its serving BS. This PDF is given in Lemma \ref{lemma_PDF_f_X}.

\begin{lemma}\label{lemma_PDF_f_X}
	The {\rm PDF} of the distance $X$ between a typical {\rm CUE} and its serving {\rm BS} represented by $f_{X}(x)$ is
	\begin{equation}\label{PDF_f_X}
		f_{X}(x)=\frac{2 \pi \lambda_{\rm b}}{\mathcal{P}} x \exp(-\frac{\gamma x^{\alpha}}{k P_{\rm b}} - \pi \lambda_{\rm b} x^2).
	\end{equation}
\end{lemma}
\begin{proof}
	The event $X > x$ is equal to $r_{\rm c} > x$ conditioned on the {\rm UE} choosing cellular mode, i.e. we have,
	\begin{flalign}\label{formula_PDF_X}
		& \mathbb{P}[X > x] = \mathbb{P}[r_{\rm c} > x | \mathcal{X} = {\rm c}] = \frac{\mathbb{P}[r_{\rm c} > x, \mathcal{X} = {\rm c}]}{\mathbb{P}[\mathcal{X}={\rm c}]},
	\end{flalign}
	where $\mathbb{P}[\mathcal{X}={\rm c}] = \mathcal{P}$ and for $\mathbb{P}[r_{\rm c} > x, \mathcal{X} = {\rm c}]$ we have,
	\begin{flalign} \label{PDF_cellular_2}
		\mathbb{P}[r_{\rm c} > x, \mathcal{X} = {\rm c}] & = \mathbb{P}[r_{\rm c} > x, k P_{\rm b} h_{\rm c} r_{\rm c}^{-\alpha} > \gamma] \notag & \\
		& = \int_{x}^{\infty} \mathbb{P}[h > \frac{\gamma r_{\rm c}^{\alpha}}{k P_{\rm b}}] f_{r_{\rm c}}(r){\rm d}r \notag & \\ 
		& = 2 \pi \lambda_{\rm b} \int_{x}^{\infty} r \exp(-\frac{\gamma r_{\rm c} ^{\alpha}}{k P_{\rm b}}) \exp(-\pi \lambda_{\rm b} r_{\rm c}^{\alpha}) {\rm d}r.
	\end{flalign} 
	Plugging (\ref{PDF_cellular_2}) into (\ref{formula_PDF_X}) yields 
	\begin{align}
		\mathbb{P}[X > x] = \frac{2 \pi \lambda_{\rm b}}{\mathcal{P}} \int_{x}^{\infty} r \exp(-\frac{\gamma r_{\rm c} ^{\alpha}}{k P_{\rm b}}) \exp\{-\pi \lambda_{\rm b} r_{\rm c}^{\alpha}\} {\rm d}r.
	\end{align}
	The cumulative density function ({\rm CDF}) of $X$ is defined by $F_X(x) = 1 - \mathbb{P}[X > x]$. Now, by taking derivation of the $F_X(x)$ over $x$, the intended {\rm PDF} denoted by $f_X(x)$ is obtained as
	\begin{align}
		f_{X}(x)=\frac{dF_X(x)}{dx} = \frac{2 \pi \lambda_{\rm b}}{\mathcal{P}} x \exp(-\frac{\gamma x^{\alpha}}{k P_{\rm b}} - \pi \lambda_{\rm b} x^2). \notag
	\end{align}
\end{proof}
In the following proposition, the coverage probability for a typical UE is computed.

\begin{prop}\label{Proposition_coverage_probability_cellular}
With an overlaied {\rm cellular-D2D}, the coverage probability of a typical {\rm CUE} associated to its nearest {\rm BS} is
\begin{align}\label{Coverage_probability_cellular}
{\mathcal{C}_{\rm c}} = & {\mathbb{E}_r}[\mathbb{P}[\text{SINR}_{\rm c}(r) \geq \beta ]] =  \frac{2\pi {\lambda_{\rm b}}}{\mathcal{P}}
\int_0^\infty r\exp \Big(- \pi {\lambda_{\rm b}}{r^2} \Big) \notag & \\ & \times 
\exp \Big( - \frac{{\beta {\sigma ^2}{r^\alpha }}}{{{P_{\rm b}}}} - \frac{\gamma r^\alpha}{k P_{\rm b}} -2\pi {\lambda_{\rm b}}{{(\beta {r^\alpha})}^{\frac{2}{\alpha}}}L(\alpha ,\beta ) \Big) {\rm d}r,
\end{align}
where $L(\beta,\alpha) = \int_{\beta^{-\frac{1}{\alpha}}}^{\infty} \frac{u}{1 + u^{\alpha}}{\rm d}u$. \\
For $\alpha = 4$ as a special case, we have \\
$
\mathcal{C}_{\rm c} = \frac{\pi \lambda_{\rm b}}{\mathcal{P}} \sqrt{\frac{\pi k P_{\rm b}}{4 (k \beta \sigma ^2 + \gamma)}} \exp\bigg(\frac{k P_{\rm b} \pi ^2 \lambda_{\rm b} ^2}{4 (k \beta \sigma^2 + \gamma)}\Big(\sqrt{\beta} \arctan{\sqrt{\beta}} + 1 \Big)^2 \bigg) \bigg[1 - \Phi \Big(\pi \lambda_{\rm b} (1+\sqrt{\beta} \arctan{\sqrt{\beta})} \sqrt{\frac{k P_{\rm b}}{4 (k \beta \sigma^2 + \gamma)}} \Big) \bigg] .
$
\end{prop}
\begin{proof}
%See Appendix B.
The proof is similar to other works done in \textbf{PPP}, e.g. \cite[Theorem 2]{B} (omitted due to page limitation).
\end{proof}
Furthermore, the average rate of a link associated to a typical CUE is provided in the following proposition.

\begin{prop}\label{Propisition_ergodic_rate_cellular}
With an overlaied {\rm cellular-D2D}, the average rate of a typical {\rm CUE} associated to its nearest {\rm BS} is
\begin{flalign}\label{Ergodic_rate_cellular}
 \mathcal{R}_{\rm c} & = \mathbb{E}_{r_{\rm c}}\Bigg[\mathbb{E}_{\text{SINR}_{\rm c}}\bigg[\text{ln} \Big(1 + \text{SINR}_{\rm c}(r_{\rm c})\Big)\bigg] \Bigg] \notag & \\ &
= \mathbb{E}_{r_{\rm c}}\Bigg[\int_0^{\infty} \mathbb{P}\bigg[\text{ln}\Big(1 + \text{SINR}_{\rm c}(r_{\rm c})\Big) > t\bigg]{\rm d}t \Bigg] \notag & \\ & =
\frac{2 \pi \lambda_{\rm b}}{\mathcal{P}} \int_0^{\infty} r \exp\Big(- \pi \lambda_{\rm b} r^2 - \frac{\gamma r^{\alpha}}{k P_{\rm b}}\Big) \notag & \\ & \times
\bigg(\int_0^{\infty} \exp\Big(-\frac{e^t - 1}{P_{\rm b}}r^{\alpha} \sigma^2 \Big) \notag & \\ & \times
\exp \Big(- 2 \pi \lambda_{\rm b} (\frac{e^t-1}{r^{-\alpha}})^\frac{2}{\alpha} L\big(e^t - 1, \alpha\big)\Big) {\rm d}t\bigg){\rm d}r,
\end{flalign}
where $L(e^t - 1,\alpha) = \int_{(e^t-1)^{-\frac{1}{\alpha}}}^{\infty} \frac{u}{1 + u^{\alpha}}{\rm d}u$. 
\end{prop}

\subsection{D2D mode}\label{D2D_mode}
In D2D mode, the UE selects $n^{\rm{th}}$-nearest neighbor as its pair. It should be noted that in this case, the interfering nodes are not independent from each other. They are related to each other through the $n^{\rm th}$ user selected as the pair. This dependency makes it necessary to consider the distance distribution between the interfering nodes and $n^{\rm th}$ neighbor. To obtain the subsequent analysis, we first make an assumption which will simplify the analysis in the sequel. 

\begin{assumption}
	The set of the interfering nodes operating in {\rm D2D} mode, are independent from each other.
\end{assumption}
In section \ref{results}, we will justify this assumption through the simulations and will demonstrate that the ignorance of the dependency between D2D UEs does not deteriorate the validity of the analysis. In our overlaid system model, the interfering nodes are only D2D transmitters. Even though we model the UEs as a \textbf{PPP} {\boldm $\Phi_{\rm u}$}, the UEs that operate in D2D mode can not be modeled by \textbf{PPP}. For the sake of simplicity in the D2D analysis, we assume that the D2D interfering nodes constitute a \textbf{PPP} {\boldm $\Phi_{\rm d}$}. We will justify this assumption later in Sec \ref{results}.

When a UE chooses D2D mode, it can operate either in HD or FD mode. We assume that a D2D UE with the probability of $\mathcal{P}_{\textsc{\tiny HD}}(\mathcal{P}_{\textsc{\tiny FD}})$ can potentially operate in HD (FD) mode, where $\mathcal{P}_{\textsc{\tiny HD}}+\mathcal{P}_{\textsc{\tiny FD}}=1$. As a result, we have {\boldm $\Phi _{\rm d}$} = {\boldm $\Phi _{\rm HD}$} $\bigcup$ {\boldm $\Phi _{\rm FD}$} with densities of $\lambda_{\rm FD} = \lambda_{\rm d} \mathcal{P}_{\textsc{\tiny FD}}$ and $\lambda_{\rm d} \mathcal{P}_{\textsc{\tiny HD}}$ for {\boldm $\Phi _{\rm FD}$} and {\boldm $\Phi _{\rm HD}$}, respectively.
The sets {\boldm $\Phi _{\rm FD}$} and {\boldm $\Phi _{\rm HD}$} are independent from each other \cite{F}. We consider that half of HD users are transmitters and half of them are receivers. Hence, the density of transmitters for process {\boldm $\Phi _{\rm HD}$} is $\lambda_{\rm HD} = \frac{1}{2}\lambda_{\rm d} \mathcal{P}_{\textsc{\tiny HD}}$. In FD mode, all FD users are transceivers. The experienced SINR at a typical D2D UE located at distance $r_{\rm d}$ from its transmitter can be defined as

\begin{equation}
\text{SINR}_{{\rm d}}(r_{\rm d}) = \frac{P_{\rm d} h_{\rm d} r_{\rm d}^{-\alpha}}{I_{\rm HD} +I_{\rm FD} + P_{\rm d} \Delta \mathbbm{1}_{\rm FD} + \sigma^2},
\end{equation}
where $ I_{\rm HD} = \sum_{x \in \text{{\boldm $\Phi_{\rm HD}$}}} P_{\rm d} h_{\rm d} ||x||^{-\alpha}$ and $I_{\rm FD} = \sum_{x \in \text{{\boldm $\Phi_{\rm FD}$}}} P_{\rm d} h_{\rm d}||x||^{-\alpha}$. $\mathbbm{1}_{\rm FD}$ is an indicator function that the typical UE is operating in FD mode. $P_{\rm d} \Delta$ is the residual SI in transceiver with the SI cancellation factor $\Delta$, such that $ 0 \le \Delta \le 1$. Before computing the coverage probability we need to obtain the Laplace transform of interference as the key intermediate analysis. To characterize the Laplace transform of interference, it should be noted that the transmitter of interest is fixed a priori and does not participate in interference. Hence, The set of interfering nodes is divided into two subsets: $\mathcal{S}_{\rm in}$ and $\mathcal{S}_{\rm out}$, in which that $\mathcal{S}_{\rm in}$ stands for the interfering nodes located inside the ball of radius $r_{\rm d}$, and $\mathcal{S}_{\rm out}$ stands for the interfering nodes lie outside of the ball of radius $r_{\rm d}$. For the sake of simplicity in Laplace transform derivation, we make the following assumption.

\begin{assumption}\label{assump1}
	Although in {\rm \textbf{PPP}} model, we work with the density and there is no limitation on the total number of {\rm UEs}, to simplify the expressions in the sequel, we assume that the total number of {\rm UEs} is fixed and equal to $W$. We also assume that it may be possible that not all {\rm UEs} are operating in transmitting mode. So we denote the set of transmitting {\rm UEs} by $\mathcal{N}_t$ which is Poisson distributed with mean $\bar{m}$ conditioned on $|\mathcal{N}_t| \leq |W|$.
\end{assumption}

\begin{theorem}\label{Theorem_coverage_probability_D2D}
For a typical {\rm UE} that chooses {\rm HD-D2D} mode based on the mode selection scheme described in section \ref{System Model}, the coverage probability is given by
\begin{flalign}
\label{Coverage_probability_HD_D2D}
& \mathcal{C}_{{\rm HD}} = \mathbb{E}_{r_{\rm d}}[\mathbb{P}[\text{SINR}_{\rm d}(r_{\rm d}) \geq \beta]] 
= \notag & \\ & \int_0^{\infty} \frac{2}{\Gamma (n)} (\lambda_{\rm HD} \pi)^n r_{\rm d}^{2n - 1} \mathcal{L}_{I_{\rm HD}}^{{\rm HD}}(\frac{\beta}{P_{\rm d} r_{\rm d}^{-\alpha}}) \mathcal{L}_{I_{\rm FD}}^{\rm HD}(\frac{\beta}{P_{\rm d} r_{\rm d}^{-\alpha}}) \notag & \\ & \times \exp\Big(-\pi \lambda_{\rm HD} r_{\rm d}^2 - \frac{\beta \sigma^2 r_{\rm d}^{\alpha}}{P_{\rm d}}\Big) {\rm d}r_{\rm d},
 \end{flalign}
 where
  $\mathcal{L}_{I_{\rm HD}}^{\rm HD}(s)$ and $\mathcal{L}_{I_{\rm FD}}^{\rm HD}(s)$ stand for the Laplace transform of interference of {\rm HD} and {\rm FD} users on {\rm HD} users and are given in (18) and (19), respectively. Note that in (18), $f_{\rm min}=\min(k,n-1),\xi=\sum_{j = 0}^{M-1} \frac{\mathcal{Q}^j e^{-\mathcal{Q}}}{j!},
  \mathcal{Q} = (\bar{m}-1)$ and $M_{\rm HD} = W \mathcal{P}_{\textsc{\tiny HD}}$.
  
  For a typical {\rm UE} that chooses {\rm FD-D2D} mode based on the mode selection scheme described in section \ref{System Model}, the coverage probability is given by
  \begin{flalign}
  \label{Coverage_probability_FD_D2D}
  & \mathcal{C}_{{\rm FD}} = \mathbb{E}_{r_{\rm d}}[\mathbb{P}[\text{SINR}_{\rm d}(r_{\rm d}) \geq \beta]] 
  = \notag & \\ & \int_0^{\infty} \frac{2}{\Gamma (n)} (\lambda_{\rm FD} \pi)^n r_{\rm d}^{2n - 1} \mathcal{L}_{I_{\rm HD}}^{\rm FD}(\frac{\beta}{P_{\rm d} r_{\rm d}^{-\alpha}}) \mathcal{L}_{I_{\rm FD}}^{\rm FD}(\frac{\beta}{P_{\rm d} r_{\rm d}^{-\alpha}}) \notag & \\ & \times  \exp\Big(-\pi \lambda_{\rm FD} r_{\rm d}^2 - \frac{\beta \sigma^2 r_{\rm d}^{\alpha}}{P_{\rm d}}- \beta \Delta r_{\rm d}^{\alpha}\Big) {\rm d}r_{\rm d},
  \end{flalign}
  where
  $\mathcal{L}_{I_{\rm HD}}^{\rm FD}(s)$ and $\mathcal{L}_{I_{\rm FD}}^{\rm FD}(s)$ stand for the Laplace transform of interference of {\rm HD} and {\rm FD} users on {\rm FD} pairs and are given in (20) and (21), respectively. Note that in (20) $M_{\rm FD} = W \mathcal{P}_{\textsc{\tiny FD}}$.
  
\end{theorem}

\begin{figure*}[!t]
	\normalsize
	\setcounter{equation}{17}
	\begin{minipage}{1\linewidth}
	\begin{flalign}
	\label{Formula: Laplace HD_HD }
	\mathcal{L}_{I_{\rm HD}}^{\rm HD}(s) = &
	\sum_{k=0}^{M_{\rm HD}-1}\sum_{l=0}^{f_{min}} \left(\begin{array}{c}
		k \\ l
	\end{array} \right) \left(\begin{array}{c}
		\frac{n-1}{M_{\rm HD}-1}
	\end{array} \right)^l \left(\begin{array}{c}
		\frac{M_{\rm HD}-n}{M_{\rm HD}-1}
	\end{array} \right)^{(k-l)}
	\frac{1}{I_{1-\frac{n-1}{M_{\rm HD}-1}}(k-f_{\rm min},f_{\rm min}+1)}
	\notag & \\ &
	\times \bigg(\int_{0}^{r_{\rm d}} \frac{2 \pi x}{1+ s P_{\rm d} x^{-\alpha}} {\rm d}x \bigg)^l
	\bigg(\int_{r_{\rm d}}^{\infty} \frac{2 \pi x}{1+ s P_{\rm d} x^{-\alpha}} {\rm d}x \bigg)^{(n-l)}
    \frac{\big((\bar{m}_{\rm HD}-1)\big)^k e^{-(\bar{m}_{\rm HD}-1)}}{k!\xi},
	\end{flalign}
		\begin{flalign}
	\label{Formula: Laplace FD_HD }
	&\mathcal{L}_{I_{\rm FD}}^{\rm HD}(s) = \exp \bigg(-\pi \lambda_{\rm FD} (sP_{\rm d})^{\frac{2}{\alpha}} \Gamma(1+\frac{1}{\alpha}) \Gamma(1-\frac{1}{\alpha}) \bigg),&
	\end{flalign}
	\begin{flalign}
		\label{Formula: Laplace FD }
	\mathcal{L}_{I_{\rm FD}}^{\rm FD}(s) = &
	\sum_{k=0}^{M_{\rm FD}-1}\sum_{l=0}^{f_{\rm min}} \left(\begin{array}{c}
	k \\ l
	\end{array} \right) \left(\begin{array}{c}
	\frac{n-1}{M_{\rm FD}-1}
	\end{array} \right)^l \left(\begin{array}{c}
	\frac{M-n}{M_{\rm FD}-1}
	\end{array} \right)^{(k-l)}
	\frac{1}{I_{1-\frac{n-1}{M_{\rm FD}-1}}(k-f_{\rm min},f_{\rm min}+1)}
	\notag & \\ &
	\times \bigg(\int_{0}^{r_{\rm d}} \frac{2 \pi x}{1+s P_{\rm d} x^{-\alpha}} {\rm d}x \bigg)^l
	\bigg(\int_{r_d}^{\infty} \frac{2 \pi x}{1+s P_{\rm d} x^{-\alpha}} {\rm d}x \bigg)^{(n-l)}
	\frac{\big((\bar{m}_{\rm FD}-1)\big)^k e^{-(\bar{m}_{\rm FD}-1)}}{k!\xi},
	\end{flalign}
			\begin{flalign}
	\label{Formula: Laplace HD_FD }
	&\mathcal{L}_{I_{\rm HD}}^{\rm FD}(s) = \exp \bigg(-\pi \lambda_{\rm HD} (sP_{\rm d})^{\frac{2}{\alpha}} \Gamma(1+\frac{1}{\alpha}) \Gamma(1-\frac{1}{\alpha}) \bigg).&
	\end{flalign}
	
	\hrulefill
	\vspace*{4pt}
\end{minipage}
\end{figure*}

\begin{proof}
The proof is similar to \cite[Lemma 7]{H}.
\end{proof}
To simplify the equations (17) and (19), we use the assumption similar to \cite{H}, that is $W \ll \bar{m}$. Thereby we have the following corollary.

\begin{corollary}
\label{bestlink_nearest_lower_corollary}
(Best link: the nearest): The Laplace transform of the interference for {\rm HD}/{\rm FD-D2D} modes, by considering nearest link, i.e. $n=1$, is given by
\begin{flalign}
& \mathcal{L}_{I_{\rm HD}}^{\rm HD}(s) = \exp \bigg(-2 \pi \lambda_{\rm HD} \int_{r_{\rm d}}^{\infty} \frac{s P_{\rm d} x^{-\alpha}}{1+s P_{\rm d} x^{-\alpha}} x {\rm d}x \bigg),
\end{flalign}
\begin{flalign}
& \mathcal{L}_{I_{\rm FD}}^{\rm FD}(s) = \exp \bigg(-2 \pi \lambda_{\rm FD} \int_{r_{\rm d}}^{\infty} \frac{s P_{\rm d} x^{-\alpha}}{1+s P_{\rm d} x^{-\alpha}} x {\rm d}x \bigg).
\end{flalign}
For the special case $\alpha = 4$, we have
$ \mathcal{L}_{I_{\rm HD}}^{\rm HD}(\frac{\beta r_{\rm d}^4}{P_{\rm d}}) = \exp\bigg(-\pi \lambda_{\rm HD} r_{\rm d}^2 \sqrt{\beta} \arctan(\beta)\bigg)$. We have the same result for {\rm FD}.
\end{corollary}
Now, the following theorem gives the average rate of the D2D mode.
\begin{theorem}
\label{Theorem_ergodic_rate_D2D}
The average rate of a typical {\rm HD-D2D} and {\rm FD-D2D} link can be calculated respectively as,
\begin{flalign}\label{Ergodic_rate_HD_D2D}
& \mathcal{R}_{\rm HD} =
\int_{0}^{\infty} \int_{0}^{\infty} \frac{2}{\Gamma(n)} \big(\pi \lambda_{\rm HD} \big)^n r_{\rm d}^{2n - 1} \mathcal{L}_{I_{\rm HD}}^{\rm HD}\bigg(\frac{e^t - 1}{P_{\rm d} r_{\rm d}^{-\alpha}}\bigg) \notag & \\ & \times \mathcal{L}_{I_{\rm FD}}^{\rm HD}\bigg(\frac{e^t - 1}{P_{\rm d} r_{\rm d}^{-\alpha}}\bigg)  \exp\Big(-\pi \lambda_{\rm HD} r_{\rm d}^2 -(e^t - 1) r_{\rm d}^{\alpha} \frac{\sigma^2}{P_{\rm d}} \Big) {\rm d}t \text{ } {\rm d}r_{\rm d},
& \\ &
\mathcal{R}_{\rm FD} =
\int_{0}^{\infty} \int_{0}^{\infty}  \exp\Big(-\pi \lambda_{\rm FD} r_{\rm d}^2 -(e^t - 1) r_{\rm d}^{\alpha} (\frac{\sigma^2}{P_{\rm d}} +  \Delta) \Big) \times \notag & \\ & \mathcal{L}_{I_{\rm HD}}^{\rm FD}\bigg(\frac{e^t - 1}{P_{\rm d} r_{\rm d}^{-\alpha}}\bigg) \notag \mathcal{L}_{I_{\rm FD}}^{\rm FD}\bigg(\frac{e^t - 1}{P_{\rm d} r_{\rm d}^{-\alpha}}\bigg) \frac{2}{\Gamma(n)} \big(\pi \lambda_{\rm FD} \big)^n r_{\rm d}^{2n - 1}{\rm d}t \text{ } {\rm d}r_{\rm d}.
\end{flalign}
\end{theorem}

\begin{table}
	\caption{Simulation Parameters}
	\label{Parameters_table}
	\resizebox{\columnwidth}{!}{
	\begin{tabular}{|l|l|l|l|}
		\hline
		Parameter & Value & Parameter & Value\\
		\hline \hline
		$\lambda_{\rm b}$ & $[10^{-7}, 10^{-6}]$ \text{ BS}$/m^2$ & $\lambda_{\rm u}$ & $0.1$ \text{ UEs}$/m^2$ \\
		\hline
		$\gamma$ & $0$ \text{ dBm} & $\sigma ^2$ & $-96 \text{ dBm/Hz}$ \\
		\hline
		$P_{\rm b}$ & $40$ \text{ dB} & $P_{\rm d}$ & $23$ \text{ dBm}\\
		\hline
		$\alpha$ & $4$ & $\Delta$ & $10^{-5}$ \text{ dB} \\
		\hline
	\end{tabular}
}
\end{table}

\section{Numerical Results}
\label{results}
In this section, we validate the analytical results and evaluate the performance of our proposed system. The summary of the simulation parameters are shown in Table \ref{Parameters_table}. Fig. \ref{Cov_a} demonstrates the coverage probability for the cellular mode for different values of $\lambda_{\rm b}$. For different values of $\lambda_{\rm b}$, the coverage probability of cellular link does not vary. This is because each UE connects to its nearest BS after it meets the condition of BS association. Since by changing the density of BSs, the distance to the nearest BS changes as well, the achieved SINR remains fixed. Fig. \ref{Cov_b} depicts the coverage probability for D2D mode. In comparison to \cite{K}, in which the coverage probability of FD-D2D is less than HD, in our case the coverage probability of FD-D2D is more than HD-D2D case. This is due to exploiting the selection of $n^{\rm{th}}$-nearest neighbor as D2D pair. In the nearest neighbor case, even though the number of interfering nodes increases as the density of UEs grows, the distance between the transmitter and receiver becomes shorter. We can see the validity of this claim in both simulation and analytic results. It should be noted that the bias factor $k$ plays no role either in cellular coverage probability or in that of D2D links. Since in cellular case, we have considered downlink transmissions, the interfering nodes are only other BSs, thus, parameter $k$, whose effect is limited to the density of UEs, has no influence on cellular coverage probability. Any modification in density of UEs in both FD and HD also changes the distance between transmitter and receiver. This leads to sustaining the experienced SINR. Fig. \ref{Th} illustrates the system sum throughput as a function of the bias factor $k$. Total throughput in each mode is defined as $\lambda_{\mathcal{X}} \mathcal{R}_{\mathcal{X}}$. Similar to coverage probability, the average rate of the cellular and D2D links does not vary by altering the bias factor $k$ since the level of SINR remains stable. However, system sum throughput changes for both cellular and D2D modes. As mentioned before, the higher $k$ values, the less number of UEs choose D2D link, i.e. more UEs select cellular mode. As a result, trend of the system sum throughput for cellular and D2D links is reverse of one another. Although the total throughput of each mode is affected by $k$, the system sum throughput which is defined by $\mathcal{T} = \lambda_{\rm c} \mathcal{R}_c + \lambda_{\rm d} \mathcal{R}_d$ remains stable for a network consisted of both cellular and HD-D2D links, as shown in Fig. \ref{Th}. We can observe that the metric for choosing suitable bias factor $k$, could be influenced by considering other target performance goals, like system load on BSs rather than system sum throughput. However, in contrast to HD-D2D links, bias factor $k$ determines the system sum throughput in a cellular network including FD-D2D links. As can be seen in Fig. \ref{Th}, setting $k=0$, i.e. forcing FD-D2D mode for all UEs, results in maximum system sum throughput. The reason is that in FD mode, both UEs in a pair can transmit simultaneously, thereby the throughput would be nearly doubled and system sum throughput increases. In case of a pure D2D network, FD mode always outperforms its counterpart HD mode in terms of throughput. From theoretical point of view, throughput of the FD system can be doubled, providing that SI is entirely canceled out. However, in practice, it is difficult or even impossible to cancel the SI perfectly. 

\begin{figure}[h]
	\centering
	\includegraphics[width=0.5 \textwidth]{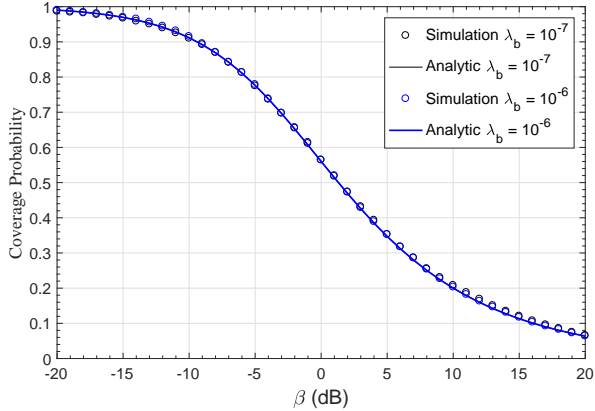}
	\caption{Coverage probability versus SINR threshold $\beta$ for cellular mode}{\vspace{-5mm}}
	\label{Cov_a}
\end{figure}

\begin{figure}[h]
	\centering
	\includegraphics[width=0.5 \textwidth]{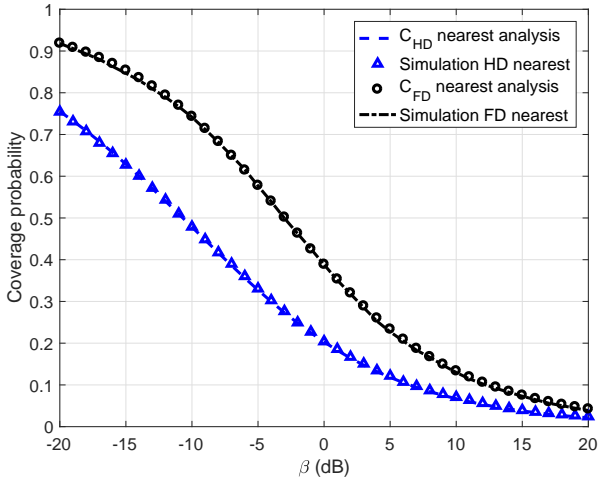}
	\caption{Coverage probability versus SINR threshold $\beta$ for D2D mode.}{\vspace{-5mm}}
	\label{Cov_b}
\end{figure}
\begin{figure}[h]
	\centering
	\includegraphics[width=0.5 \textwidth]{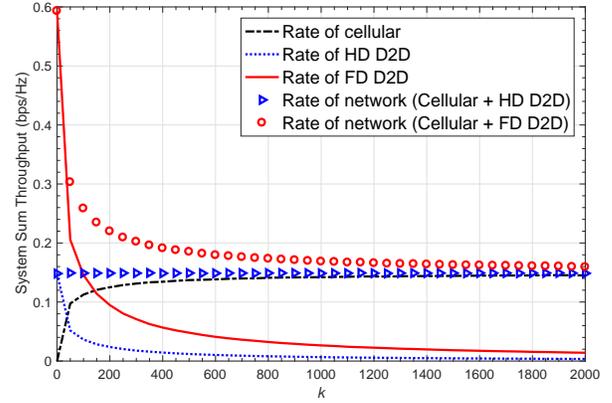}
	\caption{System sum throughput versus bias factor $k$, for different operating modes.}
	\label{Th}
\end{figure}

\section{Conclusion}
\label{Conclusion}
In this work, we proposed joint mode selection and pairing in a mixed cellular and D2D network. For the D2D mode, both HD and FD cases are analyzed. From the analytical results, several observations could be informative for system design. First, although changing the value of the bias factor leads to different densities of interferers, it does not play any role in coverage probability of D2D links due to the $n^{\rm{th}}$-nearest neighbor metric chosen for pairing. In addition, given a network with both HD and FD capabilities, we showed that unlike \cite{K}, FD coverage probability is more than that of HD. Second, in an overlay downlink cellular-HD D2D network, the total system throughput is invariant to the ratio of cellular and HD-D2D links. Precisely speaking, the values of bias factor affects only the amount of cellular or D2D links throughput individually and not the aggregate rate of them. Third, in contrast to HD-D2D links, choosing an appropriate value for the bias factor in a way that contributes to full FD-D2D network yields maximum total system throughput.

\appendices

\end{document}